\documentclass[a4paper,11pt]{scrartcl}

\usepackage[T1]{fontenc}
\usepackage[utf8]{inputenc}
\usepackage[english]{babel}
\usepackage{amsthm}
\usepackage{amsmath}
\usepackage{amssymb}
\usepackage{xspace}
\usepackage{mathtools}
\usepackage{verbatim}
\usepackage{appendix}
\usepackage{microtype}
\usepackage{hyperref}
\usepackage[labelfont=bf]{caption}

\newcommand{\cM}{\mathcal{M}}
\newcommand{\sat}{\textnormal{sat}}
\newcommand{\vbl}{\textnormal{vbl}}
\newcommand{\bestconstant}{0.617}

\theoremstyle{plain}
\newtheorem{theorem}{Theorem}
\newtheorem{lemma}[theorem]{Lemma}
\newtheorem{proposition}[theorem]{Proposition}

\title{Solving and Sampling with Many Solutions: Satisfiability and Other Hard Problems\footnote{This work started at the 2016 Gremo Workshop on Open Problems (GWOP), on June 6-10 at St. Niklausen, OW, Switzerland. }}
\author{Jean Cardinal$^1$, Jerri Nummenpalo$^2$, and Emo Welzl$^3$}
\date{
\small{
\textbf{1} Université Libre de Bruxelles (ULB), Computer Science Department, Brussels, Belgium \\ \texttt{jcardin@ulb.ac.be}\\ [0.6em] 
\textbf{2} ETH Zurich, Department of Computer Science, Zurich, Switzerland \\ \texttt{njerri@inf.ethz.ch}\\ [0.6em] 
\textbf{3} ETH Zurich, Department of Computer Science, Zurich, Switzerland \\ \texttt{emo@inf.ethz.ch}
}}

\begin{document}

\maketitle

\begin{abstract}
We investigate parameterizing hard combinatorial problems by the size of the solution set compared to all solution candidates.
Our main result is a uniform sampling algorithm for satisfying assignments of 2-CNF formulas that runs in expected time $O^*(\varepsilon^{-\bestconstant})$ where $\varepsilon$ is the fraction of assignments that are satisfying.
This improves significantly over the trivial sampling bound of expected~$\Theta^*(\varepsilon^{-1})$, and on all 
previous algorithms whenever $\varepsilon = \Omega(0.708^n)$.
We also consider algorithms for 3-SAT with an $\varepsilon$ fraction of satisfying assignments, and prove that it can be solved in $O^*(\varepsilon^{-2.27})$ deterministic time, and in $O^*(\varepsilon^{-0.936})$ randomized time. 
Finally, to further demonstrate the applicability of this framework, we also explore how similar techniques can be used for vertex cover problems.
\end{abstract}

\section{Introduction}
\label{sec:introduction}

In order to cope with the computational complexity of combinatorial optimization and satisfiability problems without sacrificing correctness guarantees, one can consider a family of instances for which a certain parameter is bounded, and analyze the complexity of algorithms as a function of this parameter. While it is now commonplace in combinatorial optimization to define the parameter as the {\em size} of a solution, we here consider computationally hard problems parameterized by the {\em number} of solutions. More precisely, we will consider satisfiability problems in which we are promised that a fraction at least $\varepsilon$ of all possible assignments are satisfying, and graph covering problems in which a fraction at least $\varepsilon$ of all vertex subsets of a certain size are solutions. 

Counting and sampling solutions to CNF formulas and more generally to CSP formulas has important practical applications. 
For example, in verification and artificial intelligence~\cite{naveh2007constraint}; and  Bayesian inference~\cite{sang2005performing}.
Recent algorithmic developments have made possible practical algorithms that can tackle industrial scale problems \cite{meel2016constrained}.

In contrast to that line of work we focus on the exact complexity of sampling, in particular to sampling solutions for 2-CNF formulas, and show that we can significantly improve on the \emph{trivial sampling algorithm} that repeatedly samples uniformly in the search space and terminates after $\varepsilon^{-1}$ steps on average.
A few previous works have also considered satisfiability problems under the promise that there are many solutions, most notably from Hirsch~\cite{hirsch1998fast}, and more recently from Kane and Watanabe~\cite{kane2013short}. 
Their focus has been on deterministic algorithms and we extend their work while also adding the consideration of randomized algorithms for $k$-SAT.

Before detailing our contributions more precisely, we briefly summarize the current state of knowledge regarding this family of questions.

\subsection{Background and previous work on satisfiability}

Hirsch~\cite{hirsch1998fast} developed a deterministic algorithm that finds a satisfying assignment for a $k$-CNF formula $F$ with an $\varepsilon$ fraction of satisfying assignments in time $O^*(\varepsilon^{-\delta_k})$ where $(\delta_k)_{k =2}^{\infty}$ is a positive increasing sequence defined by the roots of the characteristic polynomials of certain recurrence relations. The constant obtained for $k=3$ is $\delta_3 \approx 7.27$.
The main idea in his algorithm is that such formulas $F$ have \emph{short implicants} which are satisfying assignments that need to fix only few variables --- in this case only $O(\log \varepsilon^{-1})$ many --- and such assignments can be found relatively fast with a branching algorithm.
Trevisan~\cite{trevisan2004note} proposed a similar  algorithm to that of Hirsch but with an explicit running time of $O^*(\varepsilon^{-(\ln 4)k2^k})$. Although his algorithm is slightly simpler, the performance guarantees, at least for small $k$, are worse.

Kane and Watanabe~\cite{kane2013short} looked at general CNF formulas in a similar setting.
They assume that $\varepsilon \geq 2^{-n^{\delta}}$, that the number of clauses is bounded by $n^{1+\delta'}$ and that $\delta + \delta' < 1$.
Under these conditions they show that the formula has a short implicant that only fixes a linear fraction of the variables and they provide a $O^*(2^{n^\beta})$ time algorithm for finding a solution with $\beta < 1$.

Classical derandomization tools naturally apply in this context. For arbitrary CNF formulas on $n$ variables with $\varepsilon 2^n$ satisfying assignments, one can obtain a deterministic algorithm by using a pseudorandom generator that $\varepsilon$-fools depth-2 circuits. A result by De et al.~\cite{DeETT10} provides such pseudorandom generators with seed length $O\left(\log n + \log^2 \frac{m}{\varepsilon}\log\log \frac{m}{\varepsilon}\right)$. By enumerating over all seeds, we obtain a running time of $O^*\left(\left( \frac{n}{\varepsilon}\right)^{c\cdot \log {\frac{n}{\varepsilon}}}\right)$ for some constant $c$ (assuming there are $\mathrm{poly} (n)$ clauses). A recent result of Servedio and Tan improves this running time to $n^{\tilde{O}(\log\log n)^2}$ for any $\varepsilon \geq 1/\mathrm{poly}\log(n)$~\cite{ST16}.

We let \emph{Sample-2-SAT} denote the problem of sampling exactly and uniformly a satisfying assignment. 
Due to self-reducibility of satisfiability, any algorithm for the counting problem \#2-SAT can be used to solve Sample-2-SAT with only a multiplicative polynomial loss in runtime.
In fact, so far the best algorithm for Sample-2-SAT is Wahlstr{\"o}m's \#2-SAT algorithm \cite{wahlstrom2008tighter} that runs in time $O(1.238^n)$.
In contrast to the exponential time algorithms, 2-SAT can be solved in linear time with the classical algorithm of Aspvall et al.~\cite{aspvall1979linear}.
We note that while Sample-2-SAT is between 2-SAT and \#2-SAT in complexity, under the assumption $RP \not= NP$ it is not possible to uniformly or even almost uniformly sample satisfying assignments in polynomial time.
We can use a simple threefold reduction to prove this:
\begin{itemize}
\item  The constraints for an independent set in a graph can be modeled as a 2-SAT formula. Therefore a polynomial time algorithm for Sample-2-SAT would give a polynomial time algorithm for \emph{Sample-IS}. (sampling uniformly among independent sets of any size). 
The same holds for approximate versions of the problems.
\item Such sampling algorithms would yield a fully polynomial randomized approximation scheme (FPRAS) for \#IS. See for example the article of Jerrum et al.~\cite{jerrum1986random}. 
\item Lastly, such an FPRAS exists only if $RP = NP$. For details see for example the book by Jerrum~\cite[Chapter 7, Proposition 7.7]{jerrum2003counting}.
\end{itemize}

Even when relaxing Sample-2-SAT to almost uniform sampling, the best algorithm is still the one based on Wahlstr{\"o}m's counting algorithm.
This is in contrast to $k$-CNF formulas with $k \geq 3$ which have an exponential gap between exact and almost uniform sampling. 
More precisely, the gap is between exact and approximate counting.
See Schmitt and Wanka~\cite{schmitt2013exploiting} for a table of the best algorithms.

\subsection{Our results}

In Section~\ref{sec:extending_hirsch} we recall Hirsch's~\cite{hirsch1998fast} algorithm for finding a satisfying assignment for a $k$-CNF $F$ with a fraction $\varepsilon$ of satisfying assignments. We slightly generalize his analysis to also cover improved branching rules for $k$-SAT. The resulting deterministic algorithms have running times of $O^*(\varepsilon^{-\lambda_k})$ for some positive increasing sequence $(\lambda_k)_{k=2}^{\infty}$, where for instance $\lambda_3 \leq 2.27$.
We demonstrate how similar techniques can be used for finding vertex covers and we give a deterministic algorithm running in time sublinear in $\varepsilon^{-1}$ for instances of $k$-vertex cover with at least $\varepsilon \binom{n}{k}$ solutions and $k$ bounded by some fraction of $n$. 

In Section~\ref{sec:rand_alg} we prove our main result, Theorem~\ref{thm:Sample-2-SAT_epsilon}, which describes an algorithm for Sample-2-SAT that runs in expected time $O^*(\varepsilon^{-\bestconstant})$. 
It therefore improves on the algorithm based on Wahlstr{\"o}m's algorithm \cite{wahlstrom2008tighter} when $\varepsilon = \Omega(0.708^n)$, or equivalently when $F$ has $\Omega(1.415^n)$ satisfying assignments.
We leave it as an open problem to decide whether sampling solutions to 3-CNF formulas can be done in time $O^*(\varepsilon^{-\delta})$ with $\delta < 1$ and discuss why the 2-CNF case does not generalize.
In Proposition~\ref{prop:simple_randomized_3-SAT} we show how to solve 3-SAT in time $O(\varepsilon^{-0.936}(m+n))$ using similar ideas.

\subsection{Notation}

For a Boolean variable $x$ we denote its \emph{negation} by $\bar{x}$ and for a set $V$ of Boolean variables let $\overline{V}$ be the set of negated variables. 
A \emph{literal} is either a Boolean variable or its negation and in the former case we call the literal \emph{positive} and in the latter we call it \emph{negative}.
We think of a \emph{CNF formula}, or simply a \emph{formula}, $F$ over a variable set $V$ as a set $F = \{C_1,C_2,\ldots,C_m\}$ of \emph{clauses} where each clause $C_i \subset V\cup \overline{V}$ is a set of literals without both $x$ and $\bar{x}$ in the same clause for any variable $x \in V$. 
By a $k$-CNF formula and by a $(\leq k)$-CNF we denote CNF formulas in which every clause has cardinality exactly $k$ or at most $k$, respectively.
We let $\vbl(F) \subseteq V$ denote the set of variables that appear in $F$ either as a positive or negative literal. 
The \emph{empty formula} is denoted by $\{\}$ and the \emph{empty clause} by $\square$. 
An \emph{assignment} to the variables in the formula $F$ is a function $\alpha : V \rightarrow \{0,1\}$ and it is said to \emph{satisfy} $F$ if every clause $C \in F$ is satisfied, namely, if the clause contains a literal whose value is set to 1 under the assignment. 
A satisfying assignment is also called a \emph{solution}.
The empty formula is satisfied by any assignment to the variables and the empty clause by none. 
The set of all satisfying assignments of a formula $F$ over $V$ is denoted $\sat_V(F)$, and we omit the subscript $V$ when it is clear from the context.
A \emph{partial assignment} to $F$ is a function $\beta : W \rightarrow \{0,1\}$ with $W \subseteq V$ and we let $F^{[\beta]}$ be the formula over the variables $V\setminus W$ which is attained from $F$ by removing each clause of $F$ that is satisfied under $\beta$ and then removing all literals assigned to $0$ from the remaining clauses. 
If $u \in V\cup \overline{V}$ is a literal and $i \in \{0,1\}$ we let $F^{[u \mapsto i]}$ denote $F^{[\beta]}$ where $\beta$ is the partial assignment that maps only $u$ to $i$.
By \emph{unit clause reduction} we refer to the process of repeatedly setting variables to satisfy the unit clauses until finishing the process by exhausting the unit clauses or finding the empty clause.

All the logarithms are in base 2 unless noted otherwise.

\section{Deterministic algorithms and Hirsch's method}
\label{sec:extending_hirsch}

In this section we consider Hirsch's method~\cite{hirsch1998fast} for finding a satisfying assignment to a $k$-CNF formula, and extend the analysis to accommodate any branching rule.

We first briefly recall basic definitions on branching algorithms.
A \emph{complexity measure} $\mu$ is a function that assigns a nonnegative value $\mu(F)$ to every instance $F$ of some particular problem. 
Given a problem and a complexity measure $\mu$ for it, we say that an algorithm correctly solving the problem is a \emph{branching algorithm} (with respect to $\mu$) if for every instance $F$ the algorithm computes a list $(F_1,\ldots,F_t)$ of instances of the same problem, recursively solves the $F_i$'s, and finally combines the results to solve $F$. 
Finding the list $(F_1,\ldots,F_t)$ and recursively solving each of them is called a \emph{branching}.
Letting $b_i = \mu(F) - \mu(F_i)$ we call the vector $(b_1,\ldots,b_t)$ the \emph{branching vector} associated to the branching. 
Lastly, the \emph{branching number} $\tau(b_1,\ldots,b_t)$ is defined as the smallest positive solution of the equation $\sum_{i=1}^t x^{-b_i} = 1$.
If $\lambda$ is the largest branching number of any possible branching in the algorithm and $T(F)$ is the time used to find the branching and to combine the results after the recursive calls, then the running time of the algorithm can be bounded by $O(T(F)\lambda^{\mu(F)})$.

Following Hirsch~\cite{hirsch1998fast}, we consider a {\em breadth-first} version of such a branching algorithm, taking a $k$-CNF Boolean formula $F$ as input. 
We use the number of variables as a measure, and branch on partial assignments $\beta_i$, each fixing exactly $b_i$ variables.
The set $\Phi_{\ell}$ in the algorithm below eventually contains the formulas constructed from input $F$ after fixing exactly $\ell$ variables.
\begin{enumerate}
\item set $\ell \leftarrow 0$, $\Phi_0 \leftarrow \{F\}$, and $\Phi_{\ell}\leftarrow\emptyset$ for all $\ell >0$.
\item\label{bfalg:rec} if $\{\}\in \Phi_{\ell}$, then stop and return the so far fixed variables 
\item for each $F\in \Phi_{\ell}$ such that $\square \not\in F$:
\begin{enumerate}
\item find a collection of $t$ partial assignments of the form $\beta_i : W_i\to \{0,1\}$, where $W_i\subseteq \vbl(F)$
\item for each $i\in [t]$:
\begin{enumerate}
\item $\Phi_{\ell + b_i}\leftarrow\Phi_{\ell + b_i} \cup \{ F^{[\beta_i]} \}$
\end{enumerate}
\end{enumerate}
\item $\ell\leftarrow \ell +1$; if $\ell\leq n$ then go to step \ref{bfalg:rec}
\end{enumerate}

For this algorithm to be correct, the partial assignments in 3a have to of course be chosen according to a correct branching rule.
The complete collection $\Phi_{\ell}$ can be seen as a collection of nodes of the search tree of the 
recursive algorithm, and is referred to as the $\ell$th {\em floor} of the tree.
The following lemma holds~\cite{hirsch1998fast}.
\begin{lemma}
\label{lem:floor}
$|\Phi_{\ell} |\leq \lambda^{\ell}$ where $\lambda$ is the maximum branching number of the recursion tree.
\end{lemma}

The following result was proved by Hirsch in the special case of the simple
Monien-Speckenmeyer algorithm \cite{monien1985solving}, in which the branching vector 
was $(1,2,\ldots ,k)$. We generalize it to arbitrary branching vectors.

\begin{theorem}
\label{thm:branching_algorithm}
Consider a $k$-CNF formula $F$ with $n$ variables and $m$ clauses, and suppose it has
at least $\varepsilon 2^n$ satisfying assignments. 
Then any breadth-first branching algorithm for $k$-SAT
with maximum branching number $\lambda_k < 2$ runs in time $O^*(\varepsilon^{-B})$ 
on this instance, where $B := 1/(\log_{\lambda_k}2 -1)$.
\end{theorem}

\begin{proof}
After the $(\ell-1)$th step, we created all nodes of the tree in the $\ell$th floor, and all nodes 
that have a parent in the $(\ell -1)$th floor or above. There are at most $2^{n-\ell -j}$ 
assignments for each node of the $(\ell +j)$th floor. Using Lemma~\ref{lem:floor}, we have that the total number of assignments for the 
remaining nodes is at most
\begin{equation*}
\sum_{j=0}^n \lambda^{\ell +j} 2^{n-\ell -j} 
= \lambda^{\ell} 2^{n-\ell}  \sum_{j=0}^n \lambda^j 2^{-j} 
 \leq  c_{\lambda} \lambda^{\ell} 2^{n-\ell} ,
\end{equation*}
for the constant $c_{\lambda}=2/(2-\lambda )$. Note that we use $\lambda < 2$ which holds for any nontrivial branching rule.
Because the algorithm has not terminated yet, all the $\varepsilon 2^n$ assignments are still to be found and therefore we can bound $\ell$ by
\begin{eqnarray*}
c_{\lambda} \lambda^{\ell}2^{n-\ell} & \geq & \varepsilon 2^n\\
(\lambda / 2)^{\ell} & \geq & \varepsilon / c_{\lambda} \\
\ell & \leq & \log_{\lambda / 2} (\varepsilon / c_{\lambda}). \\
\end{eqnarray*}
We then have 
\begin{eqnarray*}
\lambda^{\ell} & = & \lambda^{\log_{\lambda / 2} (\varepsilon / c_{\lambda})} \\
 & = &  \lambda^{\log_{\lambda} (\varepsilon / c_{\lambda}) / \log_{\lambda} (\lambda / 2)} \\
 & = &  (\varepsilon / c_{\lambda})^{1 / \log_{\lambda} (\lambda / 2)} \\
 & = &  (c_{\lambda} / \varepsilon)^B ,
\end{eqnarray*}
and the total number of nodes is at most $\sum_{i=0}^{\ell} \lambda^i  =  O(\lambda^{\ell})$.
\end{proof}

To get concrete bounds from Theorem~\ref{thm:branching_algorithm} it remains to find good branching rules for $k$-SAT. The improved algorithm by Monien and Speckenmeyer~\cite{monien1985solving} for $k$-SAT uses the notion of autarkies and the branching vectors appearing in the algorithm are $(1)$ and $(1,2,\ldots,k-1)$ of which the latter has the worse branching number. This directly yields the following result for $k=3$.

\begin{theorem}
\label{thm:branching_algorithm_3_sat}
Given a $3$-CNF formula $F$ on $n$ variables and an $\varepsilon > 0$ with the guarantee that $|\sat(F)| \geq \varepsilon 2^n$,
one can find a satisfying assignment for $F$ in deterministic time $O^*\left(\varepsilon^{-2.27}\right)$.
\end{theorem}

\subsection{Vertex cover}
\label{sec:VC}

The technique we have seen is not unique to satisfiability but extend easily to known graph problems. 
As an example, we now consider the \emph{vertex cover problem}: given a graph $G$ and an integer $k$, 
does there exist a subset $S\in \binom{V(G)}{k}$ such that $\forall e\in E(G), e\cap S\not=\emptyset$? 
The optimization version consists of finding a smallest subset $S$ satisfying the condition. 
We consider exact algorithms, hence the problem is equivalent to the maximum independent set problem (consider $V(G)\setminus S$).
This is naturally related to the previous results on 2-SAT: the vertex cover problem can be cast as finding a minimum-weight satisfying assignment for a monotone 2-CNF formula.

We first briefly recall a standard algorithm for finding a minimum vertex cover in a graph $G$ on $n$ vertices, if one exists, in time $O^*(1.3803^n)$. First note that if the maximum degree of the graph is 2, then the problem can be solved in polynomial time. Otherwise, pick a vertex $v$ of degree at least 3, and return the minimum of $1+VC(G-v)$ and $VC(G-v-N(v))$, where $VC$ are recursive calls, and $N(v)$ is the set of neighbors of $v$ in $G$. The running time $T(n)$ obeys the recurrence $T(n)=T(n-1)+T(n-4)$, solving to the claimed bound.
We can also analyze it with respect to the size $k$ of the sought cover, yielding $T(k)=T(k-1)+T(k-3)$, solving to $1.4656^k$.
In the latter, we do not count the total number of vertices that are processed, but only those that are part of the solution.
Hence we can distinguish the branching number $\lambda$ related to the number of vertices processed and the branching number $\rho$ related to the number of vertices included in the vertex cover (equivalently, the weight of the current partial assignment). In our case, we have $\rho < 1.4656$.

We now consider instances of the vertex cover problem in which we are promised that there are at least $\varepsilon \binom{n}{k}$ vertex covers.
Given a branching algorithm, we can parse its search tree in breadth-first order, by associating with each node the number of vertices included in $S$ so far (that is, the weight of the partial assignment). We define $\Phi_{\ell}$ as the set of nodes with such value $\ell$, and call it the $\ell$th floor. The following lemma is similar to Lemma~\ref{lem:floor}.
\begin{lemma}
$|\Phi_{\ell}|\leq \rho^{\ell}$.
\end{lemma}
After generating the $\ell$th floor $\Phi_{\ell}$, there are at most $\rho^{\ell} \binom{n - \ell}{k -\ell}$ remaining covers to check.
If this is less than the total number of solutions of size $k$, we are done.
The following statement gives an upper bound on the number of levels of the tree we need to parse.

\begin{lemma}
\label{lem:height}
  Let $\ell^* := \ln (\frac 1{\varepsilon}) / \ln (\frac{n}{\rho k})$. Then for $k,n>> \ell^*$ and $k\leq n/\rho$, we have 
$$
\rho^{\ell} \binom{n - \ell}{k -\ell} \geq \varepsilon \binom{n}{k} \Rightarrow \ell \leq \ell^*.
$$
\end{lemma}

\begin{proof}
  \begin{eqnarray*}
    \rho^{\ell} \binom{n - \ell}{k -\ell} & \geq & \varepsilon \binom{n }{k}\\
    \rho^{\ell} \frac{(n - \ell)!}{(k -\ell)! (n-k)!} & \geq & \varepsilon \frac{n!}{k! (n-k)!} \\
    \rho^{\ell} \frac{(n - \ell)!}{(k -\ell)!} & \geq & \varepsilon \frac{n!}{k!} \\
    \rho^{\ell} \frac{k!}{(k - \ell)!} & \geq & \varepsilon \frac{n!}{(n-\ell)!}
\end{eqnarray*}
For $k, n$ sufficiently large, this holds whenever
\begin{eqnarray*}
    \rho^{\ell} k^{\ell} & \geq & \varepsilon n^{\ell} \\
    (\rho k /n)^{\ell} & \geq & \varepsilon \\
    \ell \ln (\rho k / n) & \geq & \ln \varepsilon \\
    \ell & \leq & \ell^* ,
    \end{eqnarray*}
where the last line uses the assumption that $k<n/\rho$.
\end{proof}

For $n$ large enough, Lemma~\ref{lem:height} implies that if $\ell > \ell^*$ then the number of remaining solutions is smaller than the promised number $\varepsilon \binom{n}{k}$, and either we have found one already, or greedily completing any partial solution leads to a solution.
Hence the running time is within a linear factor of $\rho^{\ell^*}$, which simplifies as follows.

\begin{theorem}
Given a Vertex Cover instance composed of a graph $G$ on $n$ vertices, a number $k<n/\rho$, and
an $\varepsilon > 0$ with the guarantee that $G$ has at least $\varepsilon \binom{n}{k}$ vertex covers of size $k$,
one can find such a vertex cover in deterministic time 
$$
O^*\left( \varepsilon^{-\frac{\log \rho}{\log (\frac n {\rho k})}}\right),
$$
where $\rho$ is the branching number of an exact branching algorithm for $k$-vertex cover.
In particular, this holds for $\rho = 1.4656$.
\end{theorem}

Note that the running time remains sublinear in $1/\varepsilon$ for all values of $k$ such that
$\frac{\log \rho}{\log (\frac n {\rho k})} < 1 \Leftrightarrow k < n / \rho^2$.
Hence for those values of $k$, and in particular when $k=o(n)$, we have a deterministic algorithm for $k$-vertex cover whose complexity improves on the trivial sampling algorithm.

\section{Randomized algorithms for Sample-2-SAT and for 3-SAT}
\label{sec:rand_alg}

In this section we present our algorithm for Sample-2-SAT with an expected running time of $O\left(\varepsilon^{-\bestconstant}(m+n)\right)$ on $2$-CNF formulas with more than $\varepsilon$ fraction of satisfying assignments. 
The parameter $\varepsilon$ does not need to be a constant and the algorithms can be easily modified so that they do not need to know $\varepsilon$ in advance.
Before stating and proving our main result we consider a warm-up algorithm that gives a weaker bound but already highlights some of the main ideas.
In the end we discuss the complications of generalizing our method to Sample-3-SAT and see how to solve 3-SAT in expected time $O\left(\varepsilon^{-0.940}(m+n)\right)$ using similar techniques as for Sample-2-SAT.

Schmitt and Wanka~\cite{schmitt2013exploiting} have used analogous ideas to approximately count the number of solutions in $k$-CNF formulas.

\subsection{A warm-up algorithm for Sample-2-SAT}
\label{sec:warm-up}

We will start with a warm-up algorithm that we then improve.
Let $F$ be a 2-CNF formula over the variable set $V$ with $n:=|V|$ and with $m$ clauses.
Let $S \subseteq F$ be a greedily chosen maximal set of variable disjoint clauses.
We make the following remarks.
\begin{itemize}
\item Any satisfying full assignment for $F$ must in particular satisfy $S$ and is therefore an extension of one of the $3^{|S|}$ partial assignments to $\vbl(S)$ that satisfy all clauses in $S$.
\item Because of maximality any partial assignment of the form $\alpha : \vbl(S) \rightarrow \{0,1\}$ has the property that $F^{[\alpha]}$ is a ($\leq 1$)-CNF.
\item Counting and sampling of solutions of a ($\leq 1$)-CNF is easily done in linear time.
\end{itemize}

The set $S$ allows us on one hand to do improved rejection sampling and on the other hand to device a branching based sampling. 
More concretely, consider the following two algorithms that use $S$.

\begin{enumerate}
\item Sample uniformly among all full assignments for $F$ that satisfy all the clauses in $S$ until finding one that satisfies $F$.
\item Go through all $3^{|S|}$ partial assignments $\alpha : \vbl(S) \rightarrow \{0,1\}$ that satisfy $S$ and for each $\alpha$ compute $A_{\alpha} := |\sat_{V \setminus \vbl(S)}(F^{[\alpha]})|$, i.e., the number of satisfying assignments in $F^{[\alpha]}$.
Then $A := \sum_{\alpha} A_{\alpha}$ is the number of satisfying assignments in $F$.
Draw one partial assignment $\alpha^*$ at random so that Pr$(\alpha^* = \alpha) = A_{\alpha} / A$.
For the remaining variables choose an assignment $\beta^* : V \setminus \vbl(S) \rightarrow \{0,1\}$ uniformly among all assignments satisfying $F^{[\alpha^*]}$.
Output the full assignment which when restricted to $\vbl(S)$ is $\alpha^*$ and when restricted to $V\setminus\vbl(S)$ is $\beta^*$.
\end{enumerate}

The correctness of the first algorithm is clear since any assignment satisfying $F$ must also satisfy $S$.
One sample can also be drawn in linear time.
Because the clauses of $S$ are variable disjoint, the pool of assignments we are sampling from has $(\frac{3}{4})^{|S|}2^n$ assignments and it contains all the at least $\varepsilon 2^n$ satisfying assignments.
Therefore the probability of one sample being satisfying is at least $(\frac{4}{3})^{|S|}\varepsilon$, implying an expected runtime of $O\left(\varepsilon^{-1}(\frac{3}{4})^{|S|}(m+n)\right)$ for the first algorithm.

We need the second algorithm to balance the first one when $|S|$ is small. 
For the correctness we observe that the partial assignments $\alpha$ partition the solution space in the sense that $A = \sum_\alpha A_{\alpha} = |\sat_V(F)|$ and a simple calculation shows that the output distribution is uniform over $\sat_V(F)$.
With the remarks made before the algorithm description we conclude that the runtime of the second algorithm is $O(3^{|S|}(m+n))$.
If space is a concern, the sampling of $\alpha^*$ can be done in linear space without storing the numbers $A_{\alpha}$ as follows:
Sample a uniform number $r$ from $\{1,\ldots,A\}$ and go through the partial assignments $\alpha$ again in the same order and output the first $\alpha$ for which the total number of assignments counted up to that point reaches at least $r$.

For any given $S$ we can choose the better of the two algorithms which gives an expected runtime guarantee of
\begin{align}
O\left(\max_{|S|} \left\{3^{|S|}, \varepsilon^{-1}\left(\frac{3}{4}\right)^{|S|}\right\} \cdot (m+n)\right) = O\left(\varepsilon^{-\log_4 3} (m+n) \right)
\end{align}
where $\log_4 3 < 0.793$.
Note that we do not need to know $\varepsilon$ in advance to get the same runtime guarantee as we can simulate running both of the algorithms in parallel until one finishes. 

\subsection{A faster algorithm for Sample-2-SAT}

In the warm-up algorithm we used the set $S$ on the one hand to reduce the size of the set of assignments we are sampling from and on the other hand we used it as a small size \emph{hitting set} for the clauses in $F$: every clause in $F$ contained at least one variable from $\vbl(S)$.
To improve we will do two things.
Firstly, we will consider more complicated independent structures that improve on both aspects above, giving us both a smaller size sampling pool and a better hitting set.
Secondly, we notice that it is not necessary to always use an exact hitting set in the counting procedure but an ``almost hitting set'' is enough. 
Namely, if some small set of variables hits almost all clauses we can count the number of solutions to the remaining relatively small $(\leq 2)$-SAT with a good exponential time algorithm for \#2-SAT.

We introduce first some notation.
For $i \in \mathbb{N}$ we call a set of clauses $S$ an \emph{$i$-star} if $|S| = i$ and if there exists a variable $x$ such that for any pair of distinct clauses $C,D \in S$ we have $\{x\} = \vbl(C)\cap \vbl(D)$. 
A \emph{star} is an $i$-star for some $i$.
For $i \geq 2$ we call the variable $x$ the \emph{center} of the star and any other variable is called a \emph{leaf}.
For 1-stars we consider both of the variables as centers and neither of them as leaves.
A star is called \emph{monotone} if the center appears as the same literal in every clause of the star.
We call a set $T$ of exactly three clauses a \emph{triangle} if every 2-element subset of $T$ is a star and $T$ is not itself a star.
Finally, we call a family $\cM$ of CNF formulas \emph{independent} if no two formulas in $\cM$ share common variables.

\begin{theorem}
\label{thm:Sample-2-SAT_epsilon}
Let $F$ be a 2-CNF formula on $n$ variables and $m$ clauses and let $\varepsilon > 0$ be such that $|\sat(F)| \geq \varepsilon 2^n$.
A uniformly random satisfying assignment for $F$ can be found in expected time $O\left(\varepsilon^{-\delta}(m+n)\right)$ where $\delta < \bestconstant$.
\end{theorem}

\begin{proof}

Let $V$ be the variable set of $F$ and let $k \geq 2 $ be a constant independent of $\varepsilon$ that we fix later.
We start by constructing a sequence $(\cM_0,\cM_1,\ldots,\cM_k)$ of $k+1$ independent families of formulas where every family consists of subformulas of $F$.

Let $\cM_0$ be any independent 1-maximal family of $1$-stars (clauses) in $F$.
That is, in addition to maximality we require further that there is no clause in the family whose removal would allow the addition of two clauses in its place.
We can find $\cM_0$ with a greedy algorithm in linear time\footnote{This is equivalent to finding a 1-maximal matching in a graph: first find a maximal matching and then find a maximal set of independent augmenting paths of length 3 and augment them.}. 

To construct $\cM_1$ from $\cM_0$, we add clauses of $F$ to the 1-stars of $\cM_0$ greedily to update them into non-monotone 2-stars or triangles while maintaining independence.
As a result $\cM_1$ is an independent family of subformulas of $F$ that consists of 1-stars, non-monotone 2-stars, and triangles and no 1-star can be turned into the other two types by adding clauses of $F$ to it without revoking independence.

For $i = 2,\ldots,k$ we construct $\cM_i$ from $\cM_{i-1}$ by greedily adding clauses of $F$ to the monotone $(i-1)$-stars to turn them into monotone $i$-stars while ensuring independence.
Since $k$ is a constant, and all since greedily adding clauses can be done in linear time, the total time taken to construct the families is $O(m+n)$.
An example of $\cM_4$ can be seen in Figure~\ref{fig:ind_family_construction}.
We describe the structural properties of the families later in the proof.

\begin{figure}[h]
    \centering
    \includegraphics[width=\textwidth]{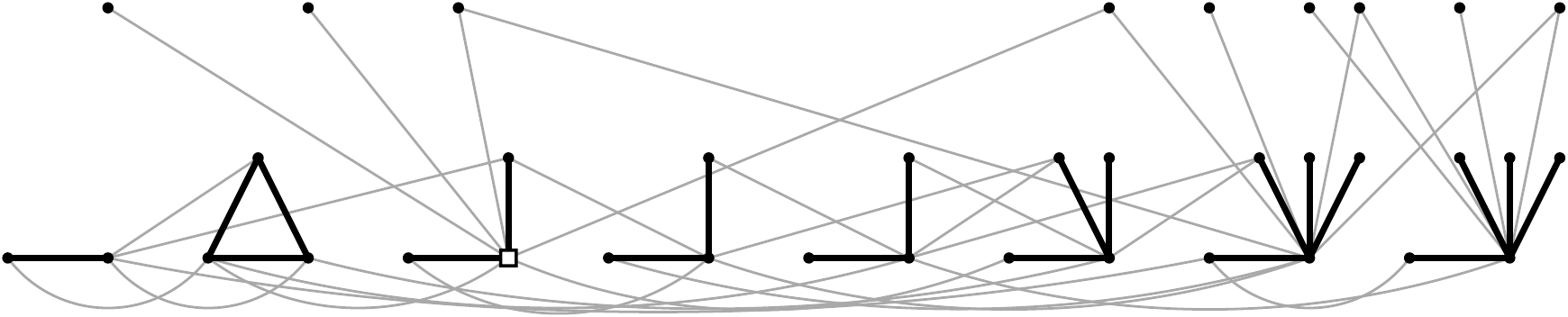}
    \caption{A possible construction of $\cM_4$ for a formula $F$ that is displayed as a graph with the variables as vertices and edges between variables appearing in the same clause. The subformulas of $F$ that make up $\cM_4$ are given by the components defined by the black bold edges. The edges that form up $\cM_0$ are the horizontal black bold edges. There is one non-monotone 2-star in $\cM_4$ and it is denoted by the square center vertex.}
    \label{fig:ind_family_construction}
\end{figure}

Analogously to the warm-up algorithm in the previous section we describe two different algorithms that both make use of the independent families we have constructed and that complement each other in terms of their running times.
The second algorithm describes in fact $k$ different algorithms, determined by the choice of a parameter $\ell \in \{1,\ldots,k\}$. 
For each $i = 1,\ldots,k$ we let $s_i$ denote the number of monotone $i$-stars in $\cM_k$.
By construction the parameter $r_i := \sum_{j=i}^k s_j$ then denotes the number of monotone $i$-stars in $M_i$.
We further let $t$ be the number of triangles and $q$ be the number of non-monotone $2$-stars in $\cM_k$, and therefore in every $\cM_i$ with $i = 1,\ldots,k$.
The two algorithms we consider are:

\begin{enumerate}
\item Sample uniformly among all full assignments for $F$ that satisfy all the clauses in $\cM_k$ until finding one that satisfies $F$.
\item Fix $\ell \in \{1,\ldots k\}$.
Define further the variable set $W := \vbl(\cM_{\ell})$ and let $W' \subseteq W$ be the set of variables of $\cM_{\ell}$ that appear in a clause of $F$ that has exactly one variable of $\cM_{\ell}$ in them.
Go through all $2^{|W'|}$ partial assignments $\alpha : W' \rightarrow \{0,1\}$ and compute $A_\alpha := |\sat_{V\setminus W'}(F^{[\alpha]})|$ by using Wahlstr{\"o}m's \#2-SAT algorithm~\cite{wahlstrom2008tighter}.
Let $A := \sum_{\alpha} A_{\alpha}$ and choose one partial assignment $\alpha^*$ at random so that Pr$(\alpha^* = \alpha) = A_{\alpha} / A$.
For the remaining variables choose an assignment $\beta^* : V \setminus W' \rightarrow \{0,1\}$ uniformly among all assignments satisfying $F^{[\alpha^*]}$. 
This can be done by branching on a variable, using Wahlstr{\"o}m's algorithm to count the number of assignments in the two branches, flipping a biased coin weighed by the counts to decide on the branch and repeating the same on the resulting formula until all variables have been set.
Output the full assignment which when restricted to $W'$ is $\alpha^*$ and when restricted to $V\setminus W'$ is $\beta^*$.
\end{enumerate}

The correctness analysis for both of these two algorithms is essentially the same as in our warm-up in Section~\ref{sec:warm-up} and it remains to discuss the running times.

Starting with the first algorithm we note that the stars and triangles in $\cM_k$ have constant size so the sampling of an assignment can be done in linear time in each iteration.
Out of the $2^{i+1}$ possible assignments to the variables in any monotone $i$-star it can be easily checked that $2^i + 1$ satisfy all the clauses in the star. 
Both for a triangle or for a non-monotone 2-star there are 8 possible assignments out of which at most 4 are satisfying.
Therefore from the independence of $\cM_k$ we know that there are at most
\begin{align}
\label{eqn:number_of_assignments}
2^{-t-q}\prod_{i=1}^{k} \left(\frac{2^i+1}{2^{i+1}}\right)^{s_i}2^n
\end{align}
full assignments to the variables in $F$ that satisfy everything in $\cM_k$.
Since $F$ has at least $\varepsilon 2^n$ satisfying assignments and the size of the universe we are sampling from is given by \eqref{eqn:number_of_assignments} we conclude that the first algorithm takes expected time
\begin{align}
\label{eqn:2-sat_more_complicated_sampling}
O\left(\varepsilon^{-1}2^{-t-q}\prod_{i=1}^{k} \left(\frac{2^{i}+1}{2^{i+1}}\right)^{s_i}(m+n)\right)
\end{align}
until returning a uniform satisfying assignment.

Consider now the runtime of the second algorithm.
This is the more intricate part of the analysis and we will make use of the structure of the families that we have set up.
It may be helpful to consider Figure~\ref{fig:ind_family_construction}.
Let $F' \in \cM_{\ell}$ be one of the subformulas in the family $\cM_{\ell}$. 
We claim that $|\vbl(F')\cap W'| \leq 1$ and that if $\vbl(F')\cap W' = \{x\}$, then $F'$ is either an $\ell$-star or a non-monotone 2-star and $x$ is the center of the star.
Towards showing the claim let $\{u,v\}$ be a clause with $\vbl(u) \in W$ and $\vbl(v) \in V\setminus W$ so that $\{u,v\}$ is a witness for $\vbl(u) \in W'$.
If $\vbl(u)$ was a leaf of a star of $\cM_{\ell}$, then we could have made $\cM_0$ larger which would contradict the 1-maximality when $\vbl(u)\in \vbl(\cM_0)$ or just maximality in the case of $\vbl(u)\not\in \vbl(\cM_0)$.
For the same reasons the variable $\vbl(u)$ can not appear in any triangle.
For any $j < \ell$ the variable $\vbl(u)$ can also not be the center of a $j$-star as otherwise we would have updated that star into a monotone ($j+1$)-star when constructing $\cM_{j+1}$ or we would have created a non-monotone $2$-star already in the beginning while constructing $\cM_1$.
The options for $\vbl(u)$ that remain are the centers of $\ell$-stars and the centers of the non-monotone 2-stars.
In the case of $\ell = 1$ we still have to argue that at most one center may appear in $W'$.
If both of the centers appeared in $W'$, it would either violate the 1-maximality of $\cM_0$ or we could have turned the 1-star into a triangle which proves the claim.
Therefore we have the bound $|W'| \leq r_{\ell} + q$.

We can observe from the argumentation above that if $\alpha : W' \rightarrow \{0,1\}$ is a partial assignment for $F$, then doing unit clause reduction on the formula $F^{[\alpha]}$ results in a $2$-CNF formula over some variable set $W_{\alpha} \subseteq W \setminus W'$. 
Computing $A_{\alpha}$ with Wahlstr{\"o}m's algorithm takes time $O(c^{|W_{\alpha}|})$ \cite{wahlstrom2008tighter}.
Therefore we want to bound $|W_{\alpha}|$ as tightly as possible.
If the assignment $\alpha$ sets the center literal of a monotone $\ell$-star to 0, then the values of the $\ell$ remaining variables in the star are determined and will be set to their required values with unit clause reduction.
For a non-monotone 2-star either assignment of the center will force the value of one of the leaves and one leaf stays undetermined.
If $\alpha$ sets $i$ of the $r_{\ell}$ literals in the centers of the monotone $\ell$-stars to 0 we get the bound
\begin{align}
\label{eqn:bound_on_remaining_variables}
|W_{\alpha}| \leq q + 3t + \ell(r_{\ell}-i)+\sum_{j = 1}^{\ell-1}(j+1)s_j.
\end{align}
Among the assignments $\alpha$ that we consider there are $\binom{r_{\ell}}{i}2^{q}$ different ones that set $i$ of the central literals of the monotone $\ell$-stars to $0$.
Using formula~\eqref{eqn:bound_on_remaining_variables} we conclude that the runtime cost of going over the assignments $\alpha$ and computing the numbers $A_{\alpha}$ is
\begin{align}
\label{eqn:2-sat_more_complicated_listing_alg}
&O\left(\sum_{i=0}^{r_{\ell}} \binom{r_{\ell}}{i} 2^{q}\cdot c^{q+3t+\ell(r_{\ell}-i)+\sum_{j = 1}^{\ell-1}(j+1)s_j}\cdot(m+n)\right) \notag \\ 
= \: &O\left(c^{3t}(2c)^{q}\left(1+c^{\ell}\right)^{r_{\ell}}\left[\prod_{j=1}^{\ell-1}c^{(j+1)s_j}\right]\cdot(m+n)\right)
\end{align}
where we used the binomial theorem.
We can again use the same trick as in the warm-up algorithm to sample $\alpha^*$ without storing all the values of $A_{\alpha}$ to keep the space requirement linear.
The running time of finding $\beta^*$ with the branching procedure takes time $O(c^{|W_{\alpha^*}|}|W_{\alpha^*}| + (m+n))$ which is subsumed by \eqref{eqn:2-sat_more_complicated_listing_alg}.

We have now one algorithm with running time given by \eqref{eqn:2-sat_more_complicated_sampling} and for any $\ell \in \{1,\ldots,k\}$ we have an algorithm with running time given by \eqref{eqn:2-sat_more_complicated_listing_alg}.
Given the sequence $(\cM_1,\ldots,\cM_k)$ we choose the algorithm with the best runtime.
To find a worst case upper bound on the runtime we look for the runtime in the form
\begin{align}
\label{eqn:worst_case_runtime}
O\left(\varepsilon^{-\delta}(m+n)\right)
\end{align}
and compute the nonnegative parameters $s_1,\ldots,s_k; t$ and $q$ that maximize the minimum of the different runtimes.
Write $\sigma_i := s_i / \log \frac{1}{\varepsilon}, \tau := t / \log \frac{1}{\varepsilon}, \rho := q / \log \frac{1}{\varepsilon}$.
By taking logarithms of the runtimes \eqref{eqn:2-sat_more_complicated_sampling}, \eqref{eqn:2-sat_more_complicated_listing_alg} and \eqref{eqn:worst_case_runtime} we can write the problem of finding $\delta$ and the worst case parameters $\sigma_i, \tau, \rho$ as the linear program
\begin{equation*}
\begin{array}{rl}
\displaystyle \max_{\delta,\,\sigma_i,\tau,\rho} \quad\delta \hfill & \\
 \text{s.t.}  \hfill  - \tau - \rho + \sum_{i = 1}^k \sigma_i  \log \left(\frac{2^i+1}{2^{i+1}}\right) &\geq \:\, \delta - 1 \\ 
                   3\tau\log c + \rho\log(2c)+ \sum_{i = 1}^{\ell-1} \sigma_i  \log \left(c^{i+1}\right) + \sum_{i = \ell}^k \sigma_i  \log \left(1+c^{\ell}\right)  &\geq \:\, \delta \quad \textnormal{for all } \ell = 1,\ldots,k \\
                 \delta,\sigma_i,\tau,\rho &\geq \:\, 0 \quad  \textnormal{for all } i = 1,\ldots,k \enspace .
\end{array}
\end{equation*}
It turns out that we only need to consider $k=7$ due to the fact that $c^{j+1} > 1+c^j$ in the integers when $j \geq 7$ which implies that the running time for higher values of $k$ no longer improves.
For $k=7$ the linear program has in the optimum $\delta < 0.61618$.
The approximate values of the other variables in the optimum are $\sigma_1 \approx 0.131, \sigma_2 \approx 0.127,, \sigma_3 \approx 0.111, \sigma_4 \approx 0.084,\sigma_5 \approx 0.051, \sigma_6 \approx 0.022, \sigma_7 \approx 0.004$ and exact values of $\tau = 0$ and $\rho = 0$.
This finishes the proof.
\end{proof}

We attempted to improve the analysis by constructing families that do not consist only of stars and triangles but the runtimes we achieved were not better.
In some sense stars seem particularly good for the efficient use of Wahlstr{\"o}m's \#2-SAT algorithm as a subroutine because the set $W'$ is not too big.
We also note that while we could consider adding the option of choosing $\ell = 0$ in the second algorithm, it is easily verified that choosing $\ell = 1$ instead gives a better performance.

\subsection{A randomized algorithm for 3-SAT}

One could say that our Sample-2-SAT algorithm works because counting and sampling solutions for a $(\leq 1)$-CNF is trivial.
Direct generalizations of our method to Sample-3-SAT do not work because the same is not true for $(\leq 2)$-CNF formulas.
Instead of solving Sample-3-SAT we apply our method for 3-SAT.

\begin{proposition}
\label{prop:simple_randomized_3-SAT}
Let $F$ be a $3$-CNF formula on $n$ variables and $m$ clauses and let $\varepsilon > 0$ be such that $|\sat(F)| \geq \varepsilon 2^n$.
A satisfying assignment for $F$ can be found in expected time $O\left(\varepsilon^{-\log_{8} 7}(m+n)\right)$.
\end{proposition}

\begin{proof}
Let $S$ be a maximal set of variable disjoint clauses in $F$.
Either sample among those assignments that satisfy $S$ until finding a satisfying assignment or go through all the $7^{|S|}$ partial assignments to $\vbl(S)$ and check the satisfiability of the resulting $(\leq 2)$-CNF.

Checking through the partial assignments takes time $O(7^{|S|}\cdot(m+n))$ because each of the $7^{|S|}$ instances of $(\leq 2)$-SAT can be solved in linear time \cite{aspvall1979linear}. The rejections sampling takes expected time $O\left(\varepsilon^{-1}\left(\frac{7}{8}\right)^{|S|}(m+n)\right)$
because we are sampling from a pool of $\left(\frac{7}{8}\right)^{|S|}2^n$ assignments that contain all the at least $\varepsilon 2^n$ many satisfying assignments. 
Choosing always the better of the two methods, depending on $|S|$, gives a worst case running time of $O\left(\varepsilon^{-\log_{8} 7}(m+n)\right)$.
\end{proof}

Proposition~\ref{prop:simple_randomized_3-SAT} gives an algorithm that works for any $\varepsilon$, but there exist better algorithms for certain ranges of $\varepsilon$.
The PPSZ algorithm for 3-SAT runs in expected time $O(1.308^n)$ \cite{hertli20143} which is faster in the case that $\varepsilon = O(0.750^n)$. 
It is also possible to analyze Sch{\"o}ning's algorithm~\cite{schoning2002probabilistic} for 3-SAT to get a dependence on $\varepsilon$ by using an isoperimetric inequality for the hypercube by Frankl and F{\"u}redi~\cite{frankl1981short}.
The computation can be found in Appendix~\ref{appendix:Schoening}.
The runtime guarantee that results is $O\left(\left(\frac{4}{3}\cdot 2^{-H^{-1}(\delta)}\right)^n\right)$ in expectation where $\delta$ is the solution to $\varepsilon = 2^{(\delta-1)n}$ and where $H : (0,1/2] \rightarrow (0,1]$ is the bijective \emph{binary entropy function} defined by $H(x) = -x\log_2(x) - (1-x)\log_2(1-x)$.
The range where Sch{\"o}ning's algorithm is better than Proposition~\ref{prop:simple_randomized_3-SAT} is when $\varepsilon = O(0.929^n)$.

\section{Conclusion}
\label{sec:conclusion}

An interesting open problem is whether Sample-3-SAT can be solved time $O^*\left(\varepsilon^{-\delta}\right)$ for some $\delta < 1$.
Similarly, can we achieve such a running time for 3-SAT with a deterministic algorithm? 


We also believe that parameterizing by the number of solutions should be a fruitful approach to other problems besides satisfiability or vertex cover.

\subparagraph*{Acknowledgments} 
We would like to thank Noga Alon and József Solymosi for discussions on the problem. We also thank the reviewers of IPEC 2017 for valuable remarks that improved the exposition.

\bibliographystyle{plainurl}
\bibliography{references}

\newpage
\appendix

\section{Analysis of Sch{\"o}ning's algorithm with many assignments}
\label{appendix:Schoening}

Sch{\"o}ning's algorithm for $k$-SAT is as follows. 
Start by picking random assignment and as long as there are unsatisfied clauses pick one and flip the value of a random literal in the clause. 
Keep on flipping $3n$ times and if no satisfied assignment has been found, restart the process from a new random assignment. 
Sch{\"o}ning\cite{schoning2002probabilistic} showed that if $\alpha^{*}$ is some satisfying assignment of $F$ and if we choose $\alpha$ as our initial random assignment, then the probability that we find a satisfying assignment within the $3n$ steps is at least 
\begin{align}
\label{eqn:schoening_success_probability}
(1/(k-1))^{ d_{H}(\alpha^{*},\alpha)}
\end{align}
where $d_H(\alpha^{*},\alpha)$ is the Hamming distance of $\alpha^{*}$ and $\alpha$.
As the distance of a random assignment to a fixed satisfying assignment is binomially distributed, the probability of finding a satisfying assignment in one iteration before restarting is at least
\begin{align*}
\frac{1}{2^n}\sum_{j=0}^{n}\binom{n}{i}\left(\frac{1}{k-1}\right)^j = \left( \frac{k}{2(k-1)} \right)^n \enspace .
\end{align*}
Therefore, with the restarts, we see that Sch{\"o}ning's algorithm runs in expected time
\begin{align*}
O^{*}\left(\left( \frac{2(k-1)}{k} \right)^n\right)
\end{align*}
when $F$ is satisfiable. 
When there are many satisfying assignments we need to be able to compute or approximate the distribution of the distance of the initial random assignment to its closest satisfying assignment. 
To do this we will use an isoperimetric inequality for the hypercube which allows us to reduce the analysis to the case where the satisfying assignments are arranged very regularly.
We use the formulation from Frankl and F{\"u}redi~\cite{frankl1981short}. 
A \emph{Hamming ball} in $\{0,1\}^n$ with center $\alpha \in \{0,1\}^n$ is a set $B \subseteq \{0,1\}^n$ such that for some $r$ we have that
\begin{align*}
\{\beta \in \{0,1\}^n \: | \: d_{H}(\beta,\alpha) \leq r\} \subseteq B \subseteq \{\beta \in \{0,1\}^n \: | \: d_{H}(\beta,\alpha) \leq r+1\}.
\end{align*}
Note that this possibly less common definition of a Hamming ball allows  for Hamming balls of any cardinality.
We call $r$ the \emph{radius} of the Hamming ball.
The cardinality of a Hamming ball of radius $r = \rho n$ for a constant $\rho \in [0,1]$ is $O^*(2^{H(\rho)n})$ where $H : [0,1] \rightarrow [0,1]$ is the binary entropy function defined for $x \in (0,1)$ by $H(x) := -x\log_2(x) -(1-x)\log_2(1-x)$ and $H(0) = H(1) = 0$.
See Chapter 10 \S 11 of the book by MacWilliams and Sloane~\cite{macwilliams1977theory} for a proof.
We define the inverse $H^{-1} : [0,1] \rightarrow [0,0.5]$ by restricting the domain of $H$ into $[0,0.5]$ on which $H$ is injective.

\begin{theorem}[Frankl and F{\"u}redi~\cite{frankl1981short}]
\label{thm:isoperimetric}
Let $A,B \subseteq \{0,1\}^n$ be sets and define
\begin{align*}
d_{H}(A,B) := \min\{d_H(\alpha,\beta) \: | \: \alpha \in A, \beta \in B\}.
\end{align*}
There are two Hamming balls $A_0$ with center $\textbf{0}$ and $B$ with center $\textbf{1}$ such that $|A| = |A_0|$ and $|B| = |B_0|$ and $d_H(A_0,B_0) \geq d_{H}(A,B)$.
\end{theorem}

We are now ready to prove the following theorem which is most likely known to experts but we find it noteworthy to write a proof.

\begin{theorem}
\label{thm:schoening_analysis}
Given a $k$-SAT instance $F$ on $n$ variables and a $\delta \in [0,H(1/k)]$ with the guarantee that $|\sat(F)| \geq 2^{\delta n}$, one can find a satisfying assignment for $F$ in expected time
\begin{align*}
    O^*\left(\left(\frac{2(k-1)}{k}\right)^{n} \cdot (k-1)^{-H^{-1}(\delta)n}\right).
\end{align*}
\end{theorem}

\begin{proof}
Let $\sigma = \sigma(\delta)$ be a constant that we decide later on, let $A$ be the set of at least $2^{\delta n}$ assignments that satisfy $F$, let $B = \{\beta \in \{0,1\}^n \:|\: d_H(\alpha,\beta) \geq \sigma n + 1 \: \forall \alpha \in A\}$, and define $A_0$ and $B_0$ with respect to $A$ and $B$ as in Theorem~\ref{thm:isoperimetric}.
Because $|A_0| = |A| \geq 2^{\delta n}$, the Hamming ball $A_0$ has radius at least $H^{-1}(\delta) n + O(\log n)$.
Define $\rho := H^{-1}(\delta)$.
We want to analyze the probability that a u.a.r. assignment $\alpha \in \{0,1\}^n$ is at most at a distance of $\sigma n$ from some point in $A$. Using basic properties of distances between two Hamming balls, that $d_{H}(A_0,B_0) \geq d_H(A,B) = \sigma n + 1$ and that $A_0$ has radius $\rho n + O(\log n)$ we can compute:
\begin{align}
&\text{Pr}(d_{H}(\alpha,A) \leq \sigma n) \nonumber\\
&= \text{Pr}(\alpha \not\in B) \nonumber \\
&= \text{Pr}(\alpha \not\in B_0) \nonumber \\
&\geq \text{Pr}(d_{H}(\alpha,A_0) \leq d_H(A_0,B_0) - 1) \nonumber \\
&\geq \text{Pr}(d_{H}(\alpha,A_0) \leq \sigma n) \nonumber \\
&= \text{Pr}(d_H(\textbf{0},\alpha) \leq (\sigma + \rho)n + O(\log n)) \nonumber \\
&= \Omega^*(2^{(H(\sigma + \rho)-1)n}). \label{eqn:probability_close}
\end{align}

The probability that a random assignment is at most at a distance $\sigma n$ of $A$ is given by \eqref{eqn:probability_close} and the probability of finding a satisfying assignment when starting from such an assignment is by \eqref{eqn:schoening_success_probability} at least $(1/(k-1))^{\sigma n}$.
The inverse of the product of these two probabilities is
\begin{align*}
O^{*}\left(2^{(1-H(\sigma + \rho))n}(k-1)^{\sigma n}\right)
\end{align*}
which is a bound on the expected number of times we need a restart in Sch{\"o}ning's algorithm before finding a satisfying assignment.
We can still choose $\sigma$ and the best choice is to define $\sigma :=\max\{1/k - \rho,0\} = \max\{1/k - H^{-1}(\delta),0\}$ which makes the expected running time equal to
\begin{align}
\label{eqn:runtime_our_alg}
\begin{cases}
    O^*\left(2^{(1-H(1/k))n}(k-1)^{ (1/k- H^{-1}(\delta))n}\right)     & \quad \text{if } \delta \leq H(1/k) \\
    O^*\left(2^{(1-\delta)n}\right)  & \quad \text{otherwise.}\\
  \end{cases}
\end{align}
It is easy to check that $2^{(1-H(1/k))n}(k-1)^{n/k} = \left(\frac{2(k-1)}{k}\right)^{n}$ which finishes the proof.
\end{proof}

\end{document}